\documentclass[10pt,twocolumn]{article}

\usepackage{amsmath,amssymb}
\usepackage{amsthm}
\usepackage{enumitem}
\usepackage{geometry}
\usepackage{hyperref}
\usepackage{times}

\theoremstyle{plain}
\newtheorem{theorem}{Theorem}[section]
\newtheorem{lemma}[theorem]{Lemma}

\theoremstyle{definition}
\newtheorem{definition}[theorem]{Definition}

\theoremstyle{remark}

\title{\vspace{-1em}Do We Really Need to Read the Input?\\
An Optimality Proof for Stone Game III}
\author{Andrew Au\\[0.2em]
\small Independent Researcher\\[0.2em]
\small Corresponding author: \texttt{cshung@gmail.com}}
\date{\vspace{-1.5em}}

\begin{document}

\twocolumn[
\maketitle
\vspace{-1em}
]

\begin{abstract}
Stone Game III admits a standard backward dynamic program using \(O(n)\)
time and \(O(1)\) auxiliary space.  The upper bound is immediate, but its
optimality raises a deceptively simple question: must a correct algorithm
really inspect a linear number of input values?

For the original problem, an all-zero instance gives a short
indistinguishability proof that every position must be inspected.  This
argument appears to depend strongly on the possibility of a tie.  We show
that it does not.  Even under the promise that every input has a winner,
an adversary can force any deterministic algorithm to make
\(\Omega(n)\) inspections by combining modular move control with
indistinguishable input completions.  We also extend the argument to
positive but unbounded values, obtaining the same linear lower bound
without zeros or ties.  Together these results establish the asymptotic
optimality of the standard \(O(n)\)-time, \(O(1)\)-space solution in
several increasingly restrictive variants.
\end{abstract}

\noindent\textbf{Keywords:}
dynamic programming; lower bounds; adversary argument; input inspection;
game theory; Stone Game III.

\section{Introduction}

Stone Game III is a two-player perfect-information game introduced as
LeetCode Problem 1406 \cite{LeetCode1406}.  A row of \(n\) stones has integer
values
\[
v_0,v_1,\ldots,v_{n-1}.
\]
Alice and Bob alternate turns, with Alice moving first.  On each turn a
player removes the first one, two, or three remaining stones and adds
their values to her or his score.  Play ends when no stones remain.
Both players play optimally.  The required output is
\texttt{Alice}, \texttt{Bob}, or \texttt{Tie}, according to the final
scores.

The usual solution is a backward dynamic program, in the standard sense
of solving overlapping suffix subproblems and retaining their optimal
values \cite{CLRS}.  Only three states need to be retained, giving
\(O(n)\) time and \(O(1)\) auxiliary space.  This paper asks whether these
bounds are asymptotically optimal.

The space bound is already the smallest meaningful asymptotic bound in
the usual machine model.  The interesting issue is time.  It is tempting
to write, ``the algorithm must read the whole input,'' but that statement
is not self-evident.  Many problems can be solved without inspecting
every input position, and an adaptive algorithm may choose its probes
based on values seen earlier.

We make the input-access model explicit and prove a worst-case
\(\Omega(n)\) time lower bound.  We then remove ties from the problem and
prove the same asymptotic bound again.  The second proof explains why
linear time is inherent in the game rather than an artifact of the
\texttt{Tie} output.

As a further theoretical extension, we consider a modified domain in
which every stone value is positive.  To expose the game-theoretic
structure cleanly, this variant removes the original fixed upper bound
on stone values.  A dominant positive stone reveals the same
complement-to-four strategy and again gives an \(\Omega(n)\) lower bound.

\section{The Backward Dynamic Program}

Because the total value of all stones is fixed, maximizing one's own
score is equivalent to maximizing the difference between one's score and
the opponent's score.

\begin{definition}
Let \(D(i)\) be the maximum score difference that the player whose turn
begins at position \(i\) can force over the other player.
\end{definition}

If the current player takes \(k\) stones, their immediate gain is
\(\sum_{j=i}^{i+k-1}v_j\).  The opponent then becomes the current player
at position \(i+k\) and can force advantage \(D(i+k)\).  Therefore
\[
D(i)=
\max_{\substack{1\leq k\leq 3\\i+k\leq n}}
\left(
\sum_{j=i}^{i+k-1}v_j-D(i+k)
\right),
\qquad D(n)=0.
\]

The sign of \(D(0)\) determines the answer.  Since \(D(i)\) depends only
on \(D(i+1)\), \(D(i+2)\), and \(D(i+3)\), three variables suffice:

\begin{verbatim}
next1 = next2 = next3 = 0
for i from n - 1 down to 0:
    sum = value[i]
    best = sum - next1
    if i + 1 < n:
        sum += value[i + 1]
        best = max(best, sum - next2)
    if i + 2 < n:
        sum += value[i + 2]
        best = max(best, sum - next3)
    next1, next2, next3 = \
        best, next1, next2
\end{verbatim}

The three variables hold \(D(i+1)\), \(D(i+2)\), and \(D(i+3)\) at the
start of an iteration.  Each position performs at most three transitions.

\begin{theorem}
Stone Game III can be solved in \(O(n)\) time and \(O(1)\) auxiliary
space.
\end{theorem}

\begin{proof}
The recurrence is evaluated once at each of the \(n\) positions, with
constant work per position.  Only three dynamic-programming values, a
running sum, and a constant number of loop variables are retained.
\end{proof}

\section{Computational Model}

For a lower bound, we must specify what information an algorithm receives.
We use the random-access input model, a direct instance of deterministic
decision-tree or query complexity \cite{BuhrmanDeWolf,AroraBarak}.  The
input length \(n\) is known, and an algorithm may inspect any chosen
position \(v_i\) at unit cost.  Its choice of the next position may depend
on all values inspected so far.  Other unit-cost arithmetic and control
operations are allowed.

The lower bounds apply to deterministic algorithms that are correct on
every input in the stated domain.  An input inspection is sometimes
called a \emph{query} or \emph{probe}.  Since each inspection costs at
least one unit of time, a lower bound on inspections is also a time lower
bound.

The central tool is the standard adversary principle of
indistinguishability \cite{BuhrmanDeWolf,AroraBarak}.  If two inputs agree
at every position inspected during an execution, the algorithm has
exactly the same execution on both inputs.  It must therefore return the
same result on both.

\section{The Immediate Lower Bound}

Consider the all-zero input
\[
Z=(0,0,\ldots,0).
\]
Every move scores zero, so the correct answer is \texttt{Tie}.

\begin{theorem}
Every deterministic algorithm that correctly solves Stone Game III must,
in the worst case, inspect all \(n\) input positions.
\end{theorem}

\begin{proof}
Run the algorithm on \(Z\).  Suppose some position \(j\) is not inspected.
Construct \(Z'\) by changing only \(v_j\) from \(0\) to \(1\).

The algorithm sees the same value at every inspected position, so its
execution and output on \(Z'\) are identical to those on \(Z\).  It
therefore returns \texttt{Tie}.  However, the total value of \(Z'\) is
one.  The two final scores are integers whose sum is one, so they cannot
be equal.  Thus \(Z'\) cannot be a tie, a contradiction.

Consequently every position must be inspected on \(Z\), requiring at
least \(n\) inspections.
\end{proof}

This proves an exact \(n\)-probe lower bound, and hence an
\(\Omega(n)\) time bound.  Yet the proof uses a tie in two ways: the hard
input is a tie, and changing one value rules out a tie by parity.  It is
natural to ask whether linear time remains necessary if ties are
forbidden.

\section{The Winner-Only Variant}

Consider the promise version of Stone Game III in which the input is
guaranteed to have a winner.  The algorithm need only distinguish
\texttt{Alice} from \texttt{Bob}.  We develop a lower bound in this
stronger setting.

\subsection{Forcing a turn by grouping moves}

Suppose all stones in a prefix have value zero.  If Alice takes \(a\)
stones, where \(1\leq a\leq 3\), Bob can respond by taking \(4-a\)
stones, making the pair consume four.  Conversely, Alice may first take
one zero and then answer Bob's \(b\) with \(4-b\).  These familiar
subtraction-game responses \cite{WinningWays} control whose turn reaches
a later gadget without changing either score.  The DP calculation below
verifies the same modulo-four behavior.

\subsection{A single-stone gadget}

The gadget \(G=(0,0,0,1)\) suffices.  More strongly, only the
position containing \(1\) must remain unread; every other position may
already have been observed to be zero.

\begin{lemma}
For any position \(p\), set \(v_p=1\) and every other value to zero.
For every \(i\leq p\),
\[
D(i)=
\begin{cases}
-1,&p-i\equiv3\pmod4,\\
 1,&p-i\not\equiv3\pmod4.
\end{cases}
\]
\end{lemma}

\begin{proof}
When the indices exist, direct substitution gives
\(D(p)=D(p-1)=D(p-2)=1\) and \(D(p-3)=-1\); the cases \(p<3\) follow
from the first three equalities.  At every earlier zero,
\[
D(i)=\max\{-D(i+1),-D(i+2),-D(i+3)\}.
\]
The next three values are all \(1\) exactly when
\(p-i\equiv3\pmod4\), giving \(-1\); otherwise one is \(-1\), giving
\(1\).  Backward induction proves the formula.
\end{proof}

Thus a \(1\) at \(p\equiv3\pmod4\) is won by Bob, whereas a \(1\) at
\(p\equiv0\pmod4\) is won by Alice.  Every four consecutive positions
contain one index of each residue, so any unread block of four permits
both completions.

\section{Winner-Only Lower Bound}

We now combine the ingredients.  Since the all-zero array violates the
winner-only promise, it should not be treated as an input on which the
algorithm is required to behave correctly.  Instead, zeros are answers
given by an adversary while the algorithm probes positions.

\begin{lemma}
\label{lem:unread-block}
Any set of fewer than \(\lfloor n/4\rfloor\) inspected positions leaves
some four consecutive positions unread.
\end{lemma}

\begin{proof}
Partition the first \(4\lfloor n/4\rfloor\) positions into disjoint
blocks of four.  If fewer than \(\lfloor n/4\rfloor\) positions have
been inspected, at least one block contains no inspected position.
\end{proof}

\begin{lemma}[Adversary validity]
\label{lem:zero-adversary}
Suppose an adversary has answered every probe with zero and fewer than
\(q=\lfloor n/4\rfloor\) positions have been inspected.  The resulting
transcript is consistent with promised inputs won by Alice and by Bob.
\end{lemma}

\begin{proof}
The preceding lemma leaves four consecutive positions unread.  They
contain indices \(p_0\equiv0\pmod4\) and
\(p_3\equiv3\pmod4\).  Complete the input by placing a single \(1\) at
either \(p_0\) or \(p_3\) and zeros elsewhere.  Both completions agree
with every observed answer.  By the single-stone lemma, the first is won
by Alice and the second by Bob, so both satisfy the winner-only promise.
\end{proof}

\begin{theorem}
Every deterministic algorithm that is correct for all winner-only Stone
Game III inputs has worst-case time \(\Omega(n)\).
\end{theorem}

\begin{proof}
Suppose for contradiction that a correct algorithm has worst-case
running time \(o(n)\) on promised inputs.  For all sufficiently large
\(n\), it then makes fewer than \(q=\lfloor n/4\rfloor\) probes on every
such input.

By Lemma~\ref{lem:zero-adversary}, before the \(q\)-th probe the current
transcript has a promised completion.  Fixing such a completion shows
that the algorithm must terminate before \(q\) probes, since its
execution is identical on that input up to that point.  At termination,
the same lemma supplies two promised completions that agree with every
observed answer but have opposite winners.  The algorithm returns the
same answer on both and is therefore wrong on one, contradicting
correctness.  Hence no correct \(o(n)\)-time algorithm exists, and the
worst-case time is \(\Omega(n)\).
\end{proof}

The constant \(1/4\) is unimportant for asymptotic optimality.  The point
is that avoiding every untouched block of constant size already requires
a linear number of inspections.

\section{Positive, Unbounded Stone Values}

The preceding winner-only construction uses zeros as a score-neutral
timing region.  We now require every input value to be at least one.
This section studies a theoretical variant allowing arbitrary positive
integers, with no fixed upper bound.  The fixed bound from the programming
problem is intentionally lifted so that turn control, rather than value
range, remains the central issue.

Start with a row of ones and replace one position by a value \(M\), where
\[
M>2n+3.
\]
The value \(M\) is a dominant prize: whoever takes it necessarily wins,
because all other stones together are worth only \(n-1\).  Intuitively,
both players should therefore continue to play the complement-to-four
strategy: the primary objective is to control whose turn reaches \(M\).
The DP below makes this intuition precise and reveals a secondary effect.

\begin{lemma}
Let every stone have value one except for a dominant prize at position
\(p\).  If \(p\equiv3\pmod4\), Bob can force himself to take the prize.
If \(p\equiv0\pmod4\), Alice can force herself to take it.
\end{lemma}

\begin{proof}
Let \(F(r)\) be the DP value of a row of \(r\) ones.  Its recurrence is
\[
F(r)=\max_{1\leq k\leq\min(3,r)}\{k-F(r-k)\},
\]
and direct induction gives the period-six sequence
\[
F(r)=0,1,2,3,2,1\qquad(r\bmod6=0,1,\ldots,5).
\]
Thus every possible tail after the prize is known.  The boundary tails
\(r=0,1,2\) must be evaluated using only the moves that remain legal;
direct substitution gives the first three corresponding rows below.
For \(r\geq3\), all three transitions are available, and the period-six
formula determines the prize block.  The boundary calculations agree
with the same phase table.

Let \(r\) be the number of trailing ones.  Substituting the corresponding
tail values through the block containing \(M\) gives:
\[
{\scriptsize\setlength{\arraycolsep}{2.5pt}
\begin{array}{c|c|c}
r\bmod6&(D(p-3),D(p-2),D(p-1),D(p))&x\\
\hline
0&(3-M,M+2,M+1,M)&M-2\\
1&(2-M,M+1,M+2,M+1)&M-1\\
2&(1-M,M,M+1,M+2)&M\\
3&(2-M,M-1,M,M+1)&M-1\\
4&(3-M,M,M-1,M)&M-2\\
5&(4-M,M+1,M,M-1)&M-3
\end{array}}
\]
For sufficiently large \(M\), every row has sign pattern
\(-,+,+,+\).  More importantly, the four DP states immediately preceding
this phase-dependent block always recover the common form
\[
B(x)=(3-x,\ x+2,\ x+1,\ x),
\]
where \(x\) is given in the last column.

For example, seven preceding ones followed by effective value \(x\)
give the exact table
\[
{\setlength{\arraycolsep}{2.5pt}
\begin{array}{c|rrrrrrrr}
v_i&1&1&1&1&1&1&1&x\\
\hline
D(i)&5-x&x&x-1&x-2&3-x&x+2&x+1&x.
\end{array}}
\]
Here \(x\) is chosen sufficiently large that every displayed
\(x-\ast\) is positive and every \(\ast-x\) is negative.  These signs
identify the maximizing branch of every DP transition, so the table is
not merely heuristic.

Prepending four ones transforms it into
\[
(5-x,\ x,\ x-1,\ x-2)=B(x-2).
\]
After \(t\) such iterations the parameter is \(x-2t\).  The condition
\(x>2t+3\) makes every \(x-\ast\) term positive and every
\(\ast-x\) term negative throughout, so the maximizing branches used
above remain valid.  Thus induction preserves the
three-positive, one-negative sign pattern, while each preceding
four-stone block replaces \(x\) by \(x-2\).  Since the recovered
\(B(x)\) block starts at \(p-7\), the negative state remains aligned
with \(p-3\pmod4\); the cases \(p<7\) follow directly from the
phase table.  Hence \(p\equiv3\pmod4\) gives \(D(0)<0\), while
\(p\equiv0\pmod4\) gives \(D(0)>0\), for every tail phase.

This drop also has a direct strategic explanation.  Both players know
that the player controlling the prize must answer a move of \(a\) stones
with \(4-a\) to preserve the alignment.  The other player therefore
takes three ones, forcing the controller to take one.  The player denied
the prize gains \(3-1=2\) points per four-stone block, exactly matching
the DP decrement.
\end{proof}

This lemma supplies opposite winners at two residue classes.  Any four
consecutive unread positions contain one index \(p_3\equiv3\pmod4\) and
one index \(p_0\equiv0\pmod4\).  Answer every probe with one.  Inside
such an unread run, create either completion
\[
\begin{aligned}
v_{p_3}&=M &&\text{and all other values are one},\\
v_{p_0}&=M &&\text{and all other values are one}.
\end{aligned}
\]
The first input is won by Bob and the second by Alice.  Both contain only
positive values, both have a definite winner, and both agree with every
answer returned to the algorithm.

\begin{theorem}
Stone Game III restricted to arbitrary positive integer values and
promised to have a winner still requires \(\Omega(n)\) worst-case time.
\end{theorem}

\begin{proof}
Let \(q=\lfloor n/4\rfloor\), and suppose for contradiction that every
promised input requires fewer than \(q\) probes.  Answer every probe with
one.  Before the \(q\)-th probe, Lemma~\ref{lem:unread-block} leaves an
unread block of four, so the transcript has a promised dominant-prize
completion.  Fixing that completion shows that the algorithm must
terminate before \(q\) probes, since its execution is identical on that
input up to that point.

At termination, the unread block supplies the two dominant-prize
completions above.  They agree with every observed answer but have
opposite winners, so the algorithm is wrong on one of them.  Therefore
some promised input requires at least \(\lfloor n/4\rfloor\) probes,
which is \(\Omega(n)\).
\end{proof}

The two-point decrement explains why the construction lets \(M\) grow
with \(n\); a fixed finite value range is not claimed here.

\section{Discussion and Conclusion}

The first lower bound formalizes the common intuition that an algorithm
must read its input.  That intuition should be treated with caution.
Having \(n\) input values does not, by itself, imply an \(\Omega(n)\)
running time: an output may be insensitive to many positions, or a
promise may rule out troublesome completions.  ``The whole input must be
read'' is a claim requiring proof, not a lower-bound argument by itself.

What makes such a claim rigorous is indistinguishability.  After too few
probes, one must exhibit two valid inputs that agree at every observed
position but require different outputs.  The all-zero argument does this
with one unread position.  The winner-only arguments require more care:
they reserve enough unread space to construct completions with opposite
winners.

The winner-only proof also exposes the game's structure: complementary
moves control who receives a single nonzero stone, while an unread block
supplies either alignment.  The positive extension replaces that stone
with a dominant prize and reveals the two-point cost of preserving its
modulo-four alignment.  These arguments force only a constant fraction
of positions to be read, but that suffices to exclude \(o(n)\) time.

Combining the lower bound with the rolling dynamic program proves
\(\Theta(n)\) time and \(O(1)\) auxiliary space, even under the no-tie
promise; the time bound also survives for positive, unbounded values.
The broader lesson is methodological: input size merely suggests a lower
bound.  The proof must show that unread data admit valid completions
requiring different answers.

\noindent\textbf{Funding.}
This research received no specific grant from funding agencies in the
public, commercial, or not-for-profit sectors.


\begin{thebibliography}{10}
\small
\setlength{\itemsep}{0pt}
\setlength{\parskip}{0pt}

\bibitem{LeetCode1406}
LeetCode.
\newblock \emph{1406. Stone Game III}.
\newblock
\url{https://leetcode.com/problems/stone-game-iii/}.
\newblock Accessed August 4, 2026.

\bibitem{CLRS}
T.~H. Cormen, C.~E. Leiserson, R.~L. Rivest, and C.~Stein.
\newblock \emph{Introduction to Algorithms}.
\newblock 4th edition, MIT Press, 2022.

\bibitem{BuhrmanDeWolf}
H.~Buhrman and R.~de~Wolf.
\newblock Complexity measures and decision tree complexity: a survey.
\newblock \emph{Theoretical Computer Science}, 288(1):21--43, 2002.
\newblock
\url{https://doi.org/10.1016/S0304-3975(01)00144-X}.

\bibitem{AroraBarak}
S.~Arora and B.~Barak.
\newblock \emph{Computational Complexity: A Modern Approach}.
\newblock Cambridge University Press, 2009.

\bibitem{WinningWays}
E.~R. Berlekamp, J.~H. Conway, and R.~K. Guy.
\newblock \emph{Winning Ways for Your Mathematical Plays}, volume~1.
\newblock 2nd edition, A~K Peters, 2001.

\end{thebibliography}
\end{document}